\documentclass{article}
\usepackage[totalwidth=13.0cm,totalheight=20.0cm]{geometry}

\usepackage{latexsym,amsthm,amsmath,amssymb,url}

\newtheorem{lemma}{Lemma}

\newtheorem{theorem}{Theorem}

\newtheorem{proposition}{Proposition}

\newcommand{\N}{\mathbb{N}}
\newcommand{\pf}{\mathop{\mathrm{pf}}}
\newcommand{\coef}{\mathop{\mathrm{coef}}}
\newcommand{\sgn}{\mathop{\mathrm{sgn}}}

\newcommand{\name}[1]{\textsc{#1}}

\newenvironment{namedefn}[3]{
\par\addvspace{0.4\baselineskip}\fbox{%
\begin{minipage}[t]{0.9\linewidth}%
\begin{tabular}{p{18mm}p{120mm}}
    \multicolumn{2}{l}{{\name{#1}}} \\
        \textsl{Input:} & {#2} \\ 
        \textsl{Question:} & {#3} \\ 
   \end{tabular}
\end{minipage}}\par\addvspace{0.4\baselineskip}}

\begin{document}

\title{Parameterized Rural Postman Problem}

\author{
Gregory Gutin\footnote{Corresponding author, email: gutin@cs.rhul.ac.uk. Research of GG was supported by Royal Society Wolfson Research Merit Award.}, Magnus Wahlstr{\"o}m\\
\small  Royal Holloway, University of London\\[-3pt]
\small Egham, Surrey TW20 0EX, UK\\[-3pt]
\small \texttt{gutin@cs.rhul.ac.uk|Magnus.Wahlstrom@rhul.ac.uk}
\and Anders Yeo\\
\small Singapore University of Technology and Design \\[-3pt]
\small 20 Dover Drive, Singapore 138682 \\[-3pt]
\small \texttt{andersyeo@gmail.com}
}
\date{}
\maketitle
\begin{abstract}
\noindent
The Directed Rural Postman Problem (DRPP) can be formulated as follows: given a strongly connected directed multigraph $D=(V,A)$ with nonnegative integral weights on the arcs, a subset $R$ of $A$ and a nonnegative integer $\ell$, decide whether $D$ has a closed directed walk containing every arc of $R$ and of total weight at most $\ell$. Let $k$ be the number of weakly connected components in the the subgraph of $D$ induced by $R$. 
Sorge et al. (2012) ask whether the DRPP is fixed-parameter tractable (FPT) when parameterized by $k$, i.e., whether there is an algorithm of running time $O^*(f(k))$ where $f$ is a function of $k$ only and the $O^*$ notation suppresses polynomial factors.
Sorge et al. (2012) note that this question is of significant practical relevance and has been open for more than thirty years. Using an algebraic approach, we prove that DRPP has a randomized algorithm of running time $O^*(2^k)$ when $\ell$ is bounded by a polynomial in the number of vertices in $D$.
We also show that the same result holds  for the undirected version of DRPP, where $D$ is a connected undirected multigraph.

\end{abstract}

\section{Introduction}

In this paper, all walks in directed multigraphs (and their special types: trails, paths and cycles) are directed. For directed multigraphs, we mainly follow terminology and notation of \cite{BanGut}. A  walk $W$ is {\em closed} if the initial and terminal vertices of $W$ coincide. A {\em trail} is a walk without repetition of arcs; a {\em path} is a trail without repetition of vertices; a {\em cycle} is a closed trail with no repeated vertices apart from initial and terminal ones.
A directed multigraph $G$ is {\em weakly connected} ({\em strongly connected}, respectively) if there is a path between any pair of vertices in the underlying undirected graph of $G$ (there are paths in both directions between any pair of vertices of $G$, respectively). A {\em weakly connected component} of $G$ is a maximal weakly connected induced subgraph of $G$.

A closed trail in directed or undirected graph $G$ is called {\em Eulerian} if it includes all edges and vertices of $G$; a graph containing an Eulerian trail, is called Eulerian. The {\em balance} of a vertex $v$ of a  directed multigraph $H$ is the in-degree of $v$ minus the out-degree of $v$. 
It is well-known that an undirected (directed, respectively) multigraph $G$ is Eulerian if it is connected and each vertex is of even degree (weakly connected and the balance of every vertex is zero, respectively) \cite{BanGut,Pap1976}. Note that every Eulerian directed multigraph is strongly connected. For directed multigraphs, we will often use the term {\em connected} instead of {\em weakly connected}.


The {\sc Chinese Postman Problem} (CPP) can be formulated as follows: given a connected multigraph $G$ with nonnegative integral weights on the edges, find a closed walk of minimum total weight which contains each edge of $G$ at least once. 
CPP for both directed and undirected multigraphs is polynomial time solvable \cite{Pap1976}.

In this paper, we study the following generalization of {\sc Directed CPP}:

\begin{center}
 \begin{namedefn}%
   {{\sc Directed Rural Postman Problem (DRPP)}}%
   {A strongly connected directed multigraph $D=(V,A)$,  \newline a subset $R$
   of arcs of $D$,  a weight function $\omega: A \rightarrow \N,$ \newline and 
   an integer $\ell$.}%
   {Is there a closed walk on $D$ containing every arc of $R$\newline
   with the total weight at most $\omega(R)+\ell$, where $\omega(R)$\newline  is the total weight of arcs in $R$?}%
 \end{namedefn}
\end{center}
We also study the {\sc Undirected Rural Postman Problem (URPP)}, where $D$ is a connected undirected multigraph.

Practical applications of RPP include garbage collection, mail delivery and snow removal \cite{AsGo1995,Dro2000,EiGeLa1995}. Both undirected and directed cases of RPP are NP-hard by a reduction from the {\sc Hamilton Cycle Problem} \cite{LeRi1976}  (see also \cite{BeNiSoWe}). 

We will study the parameterized complexity of {\sc DRPP}. A parameterized problem $\Pi \subseteq \Sigma^*\times \mathbb{N}$ is called {\em fixed-parameter tractable (FPT)} with respect
to a parameter $k$ if $(x,k) \in \Pi$ can be decided by an algorithm of running time $f(k) |x|^{O(1)}$, where $f$ is a function only depending on $k.$ 
(For background and terminology on parameterized complexity we refer the reader to the monographs~\cite{DowneyFellows99,FlumGrohe06,Niedermeier06}.)

Consider {\sc DRPP} and let $k$ be the number of weakly connected components of $D[R]$, 
where $D[R]$ is the subgraph of $D$ induced by $R$. 
Sorge {\em et al.} \cite{SoBeNiWe2012} noted that the complexity of {\sc DRPP} parameterized by $k$ 
``is a more than thirty years open ...  question with significant practical relevance.'' 
Sorge {\em et al.} \cite{SoBeNiWe2012}  commented that ``$k$ is presumably small in a number of applications \cite{Fre1977,Fre1979}'' and Sorge \cite{Sor2013} remarked
 that in planning for snow plowing routes for Berliner Stadtreinigung, $k$ is between 3 and 5.
Lately, the question whether {\sc DRPP} parameterized by $k$ is FPT was raised in \cite{BeNiSoWe,DoMoNiWe2013,Kul,Sor2013,SoBeNiWe2011}.

Frederickson \cite{Fre1977,Fre1979} obtained a polynomial-time algorithm for {\sc DRPP} when $k$ is constant. 
However, $k$ influences the degree of the polynomial in the running time of Frederickson's algorithm.
Dorn {\em et al.} \cite{DoMoNiWe2013} proved that the {\sc DRPP} is FPT when parameterized by the number $a$ of arcs not from $R$ 
in a solution of the problem. However, $k\le a$ and according to Sorge {\em et al.} \cite{SoBeNiWe2012} ``it is reasonable to assume that $k$ is much smaller [than $a$] in practice''. 
Sorge {\em et al.} \cite{SoBeNiWe2011} proved that the {\sc DRPP} is FPT when parameterized by $k+b$, where $b$ is the sum of the absolute values of the balances of vertices in $G[R]$.

In the next section, we will prove that {\sc DRPP} parameterized by $k$ admits a randomized algorithm of running time $O^*(2^k)$ provided $\ell$ is bounded by a polynomial in the number of vertices in $D$. In fact, we prove this result for 
another problem called {\sc Eulerian Extension} which is equivalent to {\sc DRPP}.

It is likely that in many applications of {\sc DRPP} the weights are bounded by a polynomial in the number of vertices in the multigraph and so our result can be applied. Consider the following examples. H{\" o}hn {\em et al.} \cite{HoJaMe2012} introduced the following problem equivalent to a problem in scheduling and proved that the problem is NP-complete. Given a directed multigraph $D =(V ,A)$ with vertices $V\subset \mathbb{R}^+_0\times \mathbb{R}^+_0$, determine whether there exists a collection $H$ of pairs of vertices of the type $(u,v)$ with $u_i\ge v_i$, where $u=(u_1,u_2),\ v=(v_1,v_2)$, such that $D+ H$ is an Euler directed multigraph. Clearly, this problem is a special case of {\sc DRPP} with all arcs being of weight 0 and 1: set $R=A$, assign weight 0 to pairs of vertices that we can add to $H$ and weight 1 to all other pairs of vertices in $V$. Golovnev {\em et al.} \cite{GoKuMi2014} obtained a reduction from the {\sc Shortest Common Superstring} problem to {\sc DRPP} parameterized by $k$ and designed a faster exact algorithm for   {\sc Shortest Common Superstring} with bounded length strings using our main result, Theorem \ref{thm:main2}. 

Sorge {\em et al.} \cite{SoBeNiWe2012} remarked that the complexity question ``extends to the presumably harder undirected case of {\sc Rural Postman}.'' We show that the {\sc DRPP}  algorithmic result holds also for URPP.

Henceforth, for a positive integer $t$, $[t]$ will stand for $\{1,\dots ,t\}.$  In an attempt to solve DRPP,
Sorge {\em et al.} \cite{SoBeNiWe2012} introduced and extensively studied the following matching problem.

\begin{center}
 \begin{namedefn}%
   {{\sc  Conjoining Bipartite Matching (CBM)}}%
   {A bipartite graph $B$ with nonnegative weights on its  edges,  \newline
    a partition $V_1\cup \ldots \cup V_t$ of vertices of $B$,   a number $\ell$,  and \newline
    a graph $([t],F)$.}%
   {Decide whether $B$ has a perfect matching $M$ of total weight\newline at most $\ell$, such 
    that for each $ij\in F$ there is an edge in $M$\newline    with one end-vertex in $V_i$ 
 and the other in $V_j.$}%
 \end{namedefn}
\end{center}

Sorge {\em et al.}  \cite{SoBeNiWe2012} proved that CBM parameterized by $|F|$ is FPT-equivalent to \textsc{DRPP} parameterized by $k$ (i.e., if either of the two parameterized problems is FPT then so is the other one). 
In Section \ref{sec:match} we prove\footnote{This result was independently derived by Marx and Pilipczuk \cite{MaPi}, with a worse dependency on $|F|$: $O^*(2^{O(|F|)})$.} 
for completeness that the same tools apply to CBM, i.e., that CBM (as well as its natural variant, {\sc Conjoining General Matching} where the graph does not have to be bipartite) admits a randomized algorithm  of running time $O^*(2^{|F|})$ provided $\ell$ is bounded by a polynomial in $n$, the number of vertices in $B$. Clearly, by the reductions of   \cite{SoBeNiWe2012}, our result on \textsc{DRPP} parameterized by $k$ implies that  CBM parameterized by $|F|$ is randomized FPT (provided $\ell$ is bounded by a polynomial in $n$), but the reductions lead to an algorithm of running time $O^*(2^{O(|F|\log |F|)}).$ 

We conclude the paper in Section \ref{sec:dis} by stating some natural open problems.


\section{Parameterized Eulerian Extension}\label{sec:EE}

Let $G=(V,R)$ be a directed multigraph and let $\omega: V\times V \rightarrow \mathbb{N}$ be a weight function. 
A multiset $E$ over $V\times V$  is an {\em Eulerian extension (EE)} for $G$ if $G+E=(V,R + E)$ is an Eulerian directed multigraph, where $R + E$ is the union of two multisets, i.e., the number of copies of each arc $uw$ in the union is the sum of the numbers of copies of  $uw$ in $R$ and $E$. In the {\em Optimization Version} of the {\sc Eulerian Extension} problem, we are to find a minimum weight EE. However, in this paper we will deal with the decision version of this problem.
\begin{center}
 \begin{namedefn}%
   {{\sc Eulerian Extension}}%
   {A directed multigraph $G=(V,R)$ with no isolated vertices,\newline
   a weight function $\omega: V\times V \rightarrow \mathbb{N}$,  and 
   an integer $\ell$.}%
   {Is there an EE of total weight of most $\ell$?}%
 \end{namedefn}
\end{center}
It is not hard to see that {\sc DRPP} and {\sc Eulerian Extension} are equivalent. We can reduce {\sc DRPP} to {\sc Eulerian Extension} by first defining the weight function as follows: for every arc $uv\in A$, $\omega(u,v)$ is set to the minimum weight of an arc from $u$ to $v$ in $A$, for every arc $uv$ for which there is no arc from $u$ to $v$ in $A$, $\omega(u,v)$ is set to $\ell+1$. This defines a complete digraph $K_V$ on vertex set $V$.
Now replace $\omega(u,v)$ for every $(u,v)\in V\times V$ by the minimum weight of a path from $u$ to $v$ in $K_V$. Finally, delete from $V$ all vertices which are not incident to an arc of $R$.  This implies that $G$ has no isolated vertices.
We can reduce  {\sc Eulerian Extension} to {\sc DRPP} by including in the arc set of $D$ only arcs of weight at most $\ell$.

This equivalence was observed in~\cite{DoMoNiWe2013}. Note that $k$ is the number of connected components of $G$ and the reduction from {\sc DRPP} to {\sc Eulerian Extension} can be done in polynomial time in the number of vertices of $D$. 

We may assume that $G$ has the following properties (after applying polynomial-time reduction rules
given by Dorn et al. \cite{DoMoNiWe2013} and Sorge et al. \cite{SoBeNiWe2012}):
\begin{description}
\item[P1] Every vertex $v$ in $G$ has balance 1, 0 or $-1$;
\item[P2] For every ordered triple $u,v,w$ of distinct vertices of $V$, we have $\omega(u,w)\le \omega(u,v)+\omega(v,w)$;
\item[P3]  Every EE  can be partitioned into pairwise arc-disjoint cycles and paths such that
every path starts at a vertex of balance 1 and terminates at a vertex of balance $-1$, and, moreover, 
every vertex of nonzero balance is an end-vertex of a unique such path. 
\end{description}
Property (P1) is achieved as follows. Let $v$ be a vertex of balance $b>1$. 
We add an extra vertex $v'$ and the arc $vv'$ to $R$. The weight function $\omega: V\times V \rightarrow \mathbb{N}$ 
treats $v'$ the same as $v$ and sets $\omega(v,v')=\omega(v',v)=0$. Note that if $E$ is an EE in the new graph then
by contracting $v$ and $v'$ we get a solution of the same weight in $G$. Analogously an EE in $G$ can be transformed into an
EE for the new graph. So we have decreased the balance of $v$ by one, by adding a vertex $v'$ of balance $1$. 
Similarly, we can increase balances of vertices with balances smaller than $-1$, by adding vertices of balance $-1$.


To achieve (P2) we can replace the weight of each arc $uv\in V\times V$ by the weight of a minimum weight path from $u$ to $v$ 
(as in the reduction from {\sc DRPP} to {\sc Eulerian Extension} above). 

Let $X_1$ be all vertices in $G$ of balance $1$ and $X_{-1}$ be all vertices in $G$ of balance $-1$ and let $M$ be any perfect matching from $X_{-1}$ to $X_1$ in $V \times V$.
If $E$ is an EE of $G$, then $R + E$ is balanced, implying that {$M + E$ is balanced (by (P1)) and therefore $M + E$ can be decomposed into arc-disjoint cycles.
Removing $M$ from these cycles, gives us the paths and cycles in (P3). 
Note that the reductions leading to Properties (P1) and (P2) keep the number of connected components of $G$ unchanged (see also \cite{DoMoNiWe2013,SoBeNiWe2012}).

In the rest of this section, we will show how to solve the \textsc{Eulerian Extension} problem in time $O^*(2^k)$ by a randomized algorithm provided $\ell$ is bounded by a polynomial in the number of vertices in $G$.
The solution uses algebraic methods and dynamic programming; throughout, all algebraic operations are performed over fields of characteristic two. 

Let $k$ be the number
of connected components of $G$ and let $V=V_1 \cup \ldots \cup V_k$ be the partition of $V$ induced by the
connected components of $G$. For shorthand, we will refer to
this set of components as $[k]$, identifying the component $G[V_i]$ with
the index $i$.

\subsection{The ``no cycles'' transformation}

For every vertex $v \in V$, add to $G$ two new vertices $v', v''$ and arcs $vv',v'v,vv'',v''v,v'v''$. 
Set the weights of all these arcs to $0$. Set $\omega(v'',v')=0$,
and for any other arc treat $v'$ and $v''$ the same as $v$ (e.g., the arc $v'u$ for $u \notin \{v,v',v''\}$ 
gets weight $\omega(v',u)=\omega(v,u)$). 
Note that this transformation preserves Property (P2).

\begin{lemma}
\label{lemma:nocycles}
After the above transformation, there exists a minimum weight EE which can
be partitioned into paths with initial vertex of balance 1 and terminal vertex of balance $-1$ and in which
every vertex of nonzero balance is an end-vertex of a unique such path.
\end{lemma}
\begin{proof}
Let $G'$ and $G''$ denote $G$ before and after the above transformation. Let $E'$ be a minimum weight EE for $G'$ satisfying Property (P3) (i.e.,
$E'$ can be partitioned into a collection $Q$ of paths and cycles, where the paths start from a vertex of balance 1 and terminate at a vertex of balance $-1$), 
and no subset of $E'$ is an EE for $G'$. Let $Q_C$ be the set of cycles in $Q$. We first argue that we can pick a set of \emph{distinct representatives} for $Q_C$, 
i.e, for every cycle $C=v_1\ldots v_pv_1$ in $Q_C$, we can select a distinct vertex $v_i$ on $C$. 

For this, first observe that by the minimality of $E'$, $G'+(E' \setminus C)$ must be non-Eulerian for every $C \in Q_C$. 
Since removing a cycle does not affect the vertex balances, we find that every cycle $C \in Q_C$ must perform some essential connectivity work, i.e., for every $C \in Q_C$ there is a pair  $u,v\in C$ of vertices such that $u$ and $v$ are in different connected components of $G'+(E' \setminus C)$.
Select one such pair $e_C=\{u,v\}$ for every $C \in Q_C$, and consider the undirected graph $H$ with edge set $L=\{e_C: C \in Q_C\}$ (the vertices of $H$ are all end-vertices of $L$). 
Clearly, $H$ is acyclic, since otherwise one of the pairs $e_C$ would not be essential after all. Thus $H$ is a forest.
By orienting every tree of $H$ away from an arbitrarily selected root, we produce a distinct representative $v$ for every edge $uv$ of $H$,
which implies a system of distinct representatives for $Q_C$. 

We now construct a minimum weight EE $E''$ for $G''$ as follows. 
First, for every $C=v_1\ldots v_p v_1\in Q_C$, let $v(C)$ be the appointed representative, say $v(C)=v_1$. 
We create a path $P$ in $G''$ from $v_1''$ to $v_1'$, by replacing the first occurrence of $v_1$ in $C$ by $v_1''$ and the second by $v_1'$.
We add $P$ to $E''$; note that the weight of $P$ for $G''$ equals the weight of $C$ for $G'$. 
Second, copy all paths $P$ of $Q$ to $E''$.
Finally, for any pair $v', v''$ of $G''$ which is still not balanced by $E''$, we add the arc $v''v'$; recall that this arc is of weight zero.
Thus $E''$ has the same weight as $E'$. 
It is easy to see that $G''+E''$ is weakly connected and that all its vertices are balanced; thus $E''$ is an EE for $G''$. 
Property (P3) for $E'$ implies the final properties of $E''$ for $G''$. 

In the other direction (to show optimality and correctness of the transformation), let $E''$ be an (arbitrary) minimum weight EE for $G''$. 
Discard every pair of vertices $v'',v'$ such that $v''v' \in E''$, and for every other arc incident on a vertex $v'$ or $v''$, transfer 
the arc to the corresponding vertex $v$. This creates an EE for $G'$, since both vertex balance and weak connectivity are preserved.
\end{proof}

The following lemma is an easy observation.

\begin{lemma}
Suppose that a minimum weight EE $E$ is partitioned into paths as above. Let $P$ be
a path of the partition, starting at a vertex $u$, terminating at a vertex
$v$, and passing through a set $I \subseteq [k]$ of connected components. Let $P'$
be a minimum weight path from $u$ to $v$, subject to the constraint that $P'$ is
incident to each component $i \in I$. Then a multiset $E'$ obtained from $E$ by replacing the arcs  of $P$ by the arcs of $P'$,
is a minimum weight EE, too.
\end{lemma}

For every vertex $u$ of balance 1 and every vertex $w$ of balance $-1$, with $u \in
V_i$ and $w \in V_j$ (where we may have $i=j$), and for every $ \{i,j\}\subseteq I \subseteq
[k]$, we let $P(u,I,w)$ represent an (arbitrary, but
fixed) minimum weight path from $u$ to $w$ having a vertex in $V_t$ for each $t\in I$.
Note that by (P2), there always exists such a path that visits each connected component exactly once and only one vertex in each connected component, unless $i=j$ 
in which case component $i$ is visited twice.

\subsection{Bipartite matching form}

We treat the resulting problem as a labelled bipartite matching problem.  
Let $U$ denote the vertices of balance 1, and $W$ the
vertices of balance $-1$. Conceptually, we will solve the problem the following way:
Let $G_2=(U \cup W, E')$ be a weighted bipartite multigraph where for every $u \in
U$ and $w \in W$, with $u \in  V_i$ and $w \in V_j$, and for every $I$ such that $ \{i,j\} \subseteq I
\subseteq [k] $, there is an edge $e_{u,I,w}$ from $u$ to
$w$, representing the path $P(u,I,w)$. The weight of $e_{u,I,w}$ equals the weight of the path
$P(u,I,w)$ if it is at most $\ell$ and, otherwise, it equals $\ell+1$.
Given a perfect matching $M$ in $G_2$, we define a
multiset $EP(M)$ of arcs as follows. Initially $EP(M)=\emptyset$, then for every edge
$e_{u,I,w}\in M$ we add the arcs of the path $P(u,I,w)$ to $EP(M)$,
creating multiple copies of some arcs, if necessary. The following is a simple observation.

\begin{lemma}\label{lem:EM}
An instance of {\sc Eulerian Extension} is positive if and only if there is a perfect matching
$M$ in the above graph such that $G+EP(M)$ is weakly connected and the total weight of $EP(M)$ is at most $\ell.$
\end{lemma}


\subsection{Sieving for weakly connected EE}

Using the determinant of a bipartite adjacency matrix, we can
enumerate perfect matchings $M$ of $G_2$. By construction, all vertices of $G+ EP(M)$ will be of balance zero
for every perfect matching $M$, however, in general $G+ EP(M)$ will not
be connected and hence not Eulerian. We will sieve for matchings
$M$ which cause $G+ EP(M)$ to be connected, by setting up a sum, computing a
large polynomial over a field of characteristic two, where every term
corresponds to a perfect matching of $G_2$, and where a term is counted an
odd number of times if and only if the corresponding matching $M$ is such
that $G+ EP(M)$ is connected. This approach was previously used by Cygan {\em et al.} \cite{CyNePiPiRoWo} 
for solving connectivity problems parameterized by treewidth (see also \cite{Bjo2010}).

For every edge $e_{u,I,w}$ of $G_2$, let
$x(e_{u,I,w})=x_{u,I,w} \cdot z^r$ where $x_{u,I,w}$ is a new indeterminate,
$z$ is a common indeterminate, and $r$ equals the weight of $e_{u,I,w}$. For every $I \subseteq [k]$, denote $V_I=\bigcup_{i \in I}
V_i$, and let $A_I$ be a matrix with rows indexed by $U \cap V_I$ and
columns indexed by $W \cap V_I$, and with entries
$
A_I(u,w)=\sum_{\{u,w\}\subseteq J\subseteq I} x(e_{u,J,w}).
$
It is well-known that the determinant of a square matrix can be computed in polynomial time \cite{Rote2001}. Thus,
given $I$ and an assignment to all variables $x$ and $z$, we can
evaluate $\det A_I$ in time $O^*(2^{|I|})$, assuming that some table of the weights of the edges $e_{u,J,w}$, $J\subseteq I$, has been prepared. 
Let
\[
Q(\bar x, z):=\sum_{I \subseteq [k]\setminus\{1\}} (\det A_I)(\det A_{[k] \setminus I}),
\]
where the determinant of a non-square matrix is 0 and $\det A_{\emptyset}=1$. 
We will show that, over a field of characteristic two, $Q(\bar x,z)$ enumerates exactly those perfect matchings $M$
for which $G+ EP(M)$ is weakly connected. 

The following lemma is immediate from the Leibniz formula for the determinant. Also note that if $A_I$ is non-square then the graph $G_{2,I}$ of the lemma has no perfect matchings. 

\begin{lemma}
\label{lemma:1}
Let $I \subseteq [k]$. Let $G_{2,I}=((U \cap V_I) \cup (W \cap V_I),E_I)$
denote the bipartite multigraph with edges $e_{u,J,w}$ from $G_2[V_I]$ for which $J \subseteq I.$
Then 
\[
\det A_I=\sum_{M_I} \prod_{e \in M_I} x(e), \mbox{ where $M_I$ ranges over perfect matchings of $G_{2,I}.$}
\]
\end{lemma}

\begin{lemma}\label{lem:only-matchings}
Every term of $Q(\bar x,z)$ equals $\prod_{e \in M} x(e)$ for some perfect
matching $M$ of $G_2$, where $x(e)$ for $e=e_{u,I,w}$ equals $x(e_{u,I,w})$. For every perfect matching $M$ of $G_2$, the sum defining $Q(\bar x,z)$ has $\prod_{e \in M} x(e)$ as a term (possibly with even multiplicity).
\end{lemma}
\begin{proof}
Let $M$ be a perfect matching of $G_2$. If $I=\emptyset$, $(\det A_I)(\det A_{[k] \setminus I})=\det A_{[k]}$ has $\prod_{e \in M} x(e)$ as a term by Lemma~\ref{lemma:1}.

Let $I \subseteq [k]\setminus \{1\}$. By Lemma~\ref{lemma:1}, each term in both $\det
A_I$ and $\det A_{[k] \setminus I}$ corresponds to a perfect matching of
$G_{2,I}$ and $G_{2,[k] \setminus I}$, respectively. 
Let $M_1$ and $M_2$ be arbitrary perfect
matchings of the respective graphs. Since $V_I$ and $V_{[k] \setminus I}$
partition $V$, we have that $M_1$ and $M_2$ are disjoint and 
$M=M_1 \cup M_2$ forms a perfect matching of $G_2$. Furthermore, 
$(\prod_{e \in M_1} x(e))(\prod_{e \in M_2} x(e))=\prod_{e \in M} x(e)$.
Thus for every selection $I$, each term of $Q$ generated equals 
$\prod_{e \in M} x(e)$ for some perfect matching $M$, and $Q(\bar x,z)$
as a whole thus enumerates perfect matchings in this sense.
\end{proof}

\begin{lemma}\label{lem:discon}
Let $M$ be a perfect matching of $G_2$ such that $G+EP(M)$ has $\rho$ connected components. 
Then the term $\prod_{e \in M} x(e)$ is enumerated $2^{\rho-1}$ times
in $Q(\bar x, z)$. 
\end{lemma}
\begin{proof}
Let $\rho=1$. Recall that $M$ is generated in the term corresponding to $I=\emptyset$ (see the proof of the previous lemma). 
We show that no other choice of $I$
generates $M$. Let $I \subseteq [k]\setminus \{1\}$ be non-empty. Since $G+EP(M)$ is
connected, there is some edge $e_{u,J,w}$, $J\subseteq I$, in $M$ such that either $u$ and
$w$ lie on different sides of the partition $(V_I, V_{[k]\setminus I})$,
or $u,w$ lie on the same side but some $j \in J$ lies on the other side.
In both cases, the edge $e_{u,J,w}$ cannot be generated in either $\det
A_I$ or $\det A_{[k] \setminus I}$. 

Now we may assume that $\rho\ge 2$.
Let $C_1, \ldots, C_{\rho}$ be the connected components of $G+EP(M)$ and observe that 
for each $i\in [k]$, the vertex set $V_i$ is contained in a single component $C_j$. 
Let $[k]=I_1 \cup \ldots \cup I_{\rho}$ denote the partition of $[k]$ according to the weak components $C_i$,
i.e., $I_i=\{j \in [k]: V_j \subseteq V(C_i)\}$; choose the numbering
so that $1 \in I_1$.  For $J \subseteq [\rho]$, let $I_J:=\bigcup_{j \in J} I_j$. 

Let $1\in J\subseteq [\rho]$.
Observe that $M$ partitions into one perfect
matching $M_1$ for $G_{2,I_J}$ and one perfect matching $M_2$ for 
$G_{2, [k] \setminus I_J}$.
Since $\prod_{e \in M_1} x(e)$ is generated in $\det A_{I_J}$, and $\prod_{e \in M_2} x(e)$ is generated in $\det
A_{[k] \setminus I_J}$, $\prod_{e \in M} x(e)$ is generated (exactly once)
in their product. By an argument similar to the one used in the first paragraph of this proof, we can see that
$M$ can be generated only as above. Since there are exactly $2^{\rho -1}$ ways to choose a subset $J$ of $[\rho]$ containing 1, we are done.
\end{proof}

Thus, we have the following:

\begin{theorem}\label{thm:Q}
Over a field of characteristic two, $Q(\bar x,z)$ enumerates exactly those perfect matchings $M$
for which $G+ EP(M)$ is weakly connected. 
\end{theorem}


\subsection{Computing $Q(\bar x, z)$ fast}


We will now show how to evaluate $Q(\bar x, z)$, given an instantiation of the variables $\bar x$ and $z$, 
in time $O^*(2^k)$. This will be achieved using dynamic programming, with the main ingredient being the
fast zeta transform described below. 

We will describe a sequence of tables. First, for each $(u,I,w)$ with $u \in U \cap V_i$, $w \in W \cap V_j$, 
and $ \{i,j\} \subseteq I \subseteq [k] $, we let $d(u,I,w)$ denote the minimum weight of a path from $u$ to $w$, 
passing via exactly the set $I$ of components.


\begin{lemma}
We can fill in all values of $d(u,I,w)$ in time $O^*(2^k)$.
\end{lemma}
\begin{proof}
We may assume that $i\neq j$ as the case $i=j$ is analogous.
As previously observed, by Property (P2) for every $(u,I,w)$ as described
there is a minimum weight path from $u$ to $w$ which contains exactly $|I|$ vertices
and hence passes each component $i \in I$ exactly once. 
Using this, we first fill in all values $d(u,\{i,j\},w)=\omega(u,w)$.
Then, in increasing order of $|I|\ge 3$, we fill in values for $d(u,I,w)$ as follows:
$$d(u,I,w)=\min\{d(u,J,v)+\omega(v,w): \{u,v\}\subseteq J\subset I, |J|=|I|-1, w\in V_{I\setminus J}\}.$$
The total time for the procedure is $O^*(2^k)$. 
\end{proof}

Given this, we may now create a table for the concrete values of $x(e_{u,I,w})=x_{u,I,w}z^{d(u,I,w)}$
using the previous table and the given values for $\bar x$ and $z$. 
The remaining task is to create the matrices $A_I$ for $I \subseteq [k]$;
recall that $A_I(u,w)=\sum_{J \subseteq I} x(e_{u,J,w})$. 
To compute $A_I(u,w)$, we may use the \emph{fast zeta transform} of Yates~\cite{Yates37},
as previously also used for exact algorithms by, e.g., Bj\"orklund et al.~\cite{BjHuKo09}.

\begin{lemma}[\cite{Yates37,BjHuKo09}]
Given a function $f: 2^N \rightarrow R$ for some ground set $N$ and ring $R$,
we may compute all values of $\hat f: 2^N \rightarrow R$ defined as
$\hat f(S) = \sum_{A \subseteq S} f(A)$ using $O^*(2^{|N|})$ ring operations. 
\end{lemma}

We note that $A_I(u,w)$
 is the zeta transform of $x(e_{u,I,w})$ by definition,
hence the matrices $A_I$, $I\subseteq [k]$,
can be precomputed from the values of $\bar x$ in
total time $O^*(2^k)$. 

Now using the definition of $Q(\bar x, z)$ and the above runtime bounds, we obtain the following:

\begin{theorem}\label{thm:Qeval}
Given an instantiation of the variables $\bar x$ and $z$, we can evaluate $Q(\bar x, z)$
in time $O^*(2^k)$. 
\end{theorem}

\subsection{Main result for Eulerian Extension}

In this subsection we will prove the main result of Section \ref{sec:EE}.

\begin{theorem}\label{thm:main2}
The \textsc{Eulerian Extension} problem with $k$ weak components in $G$ and $\ell$ bounded by a polynomial in the number of vertices in $G$
can be solved by a randomized algorithm in $O^*(2^k)$ time and space.
\end{theorem}

To prove Theorem \ref{thm:main2}, apart from Theorems \ref{thm:Q} and \ref{thm:Qeval}, we will use the following two lemmas.

\begin{lemma}\label{lem:SZ} (Schwartz-Zippel \cite{Sch1980,Zip1979}). Let $P(x_1, . . . , x_n)$ be a multivariate polynomial of total
degree at most $d$ over a field $\mathbb{F}$, and assume that $P$ is not identically zero. Pick $r_1,\dots, r_n$
uniformly at random from $\mathbb{F}$. Then Pr$[P(r_1,\ldots ,r_n)=0]\le d/|\mathbb{F}|.$
\end{lemma}

\begin{lemma}
\label{lem:polynomial interpolation}
Let $f(z)=c_r z^r + \ldots + c_0$ be a polynomial over a field of size at least $r+1$. 
For any $i \in [r]$, we can express the coefficient $c_i$ as a linear combination of $r+1$ evaluations of $f(z)$;
the expression can be found in time polynomial in $r+1$. 
\end{lemma}
\begin{proof}
Let $z_0, \ldots, z_r$ be distinct elements of the field. 
Then the values of $g(z_i)$ for $i=0, \ldots, r$ collectively define the polynomial $g(z)$ entirely. 
In fact, the coefficients $c_i$ are the solutions to a linear system $BC=D$, where $B$ is the $(r+1) \times (r+1)$ Vandermonde matrix 
based on values $z_i$ (i.e., $B(i,j)=z_i^{j}$), $C$ is a column vector with $C(i)=c_i$ (indexed from $0$ to $r$), 
and $D$ is a column vector with $D(i)=g(z_i)$ (indexed similarly). 
Since $B$ is a Vandermonde matrix over distinct values $z_i$, $B$ is non-singular~\cite{HoJo1991}, hence
we may write $C=B^{-1} D$. This equation, in turn, defines each coefficient $c_i$ as a linear combination over evaluations $g(z_i)$. 
As $B^{-1}$ can be found in polynomial time, the complexity claim follows.
\end{proof}

Now Theorem \ref{thm:main2} follows from the next lemma and Theorems \ref{thm:Q} and \ref{thm:Qeval}.

\begin{lemma}
Let $G$ be an instance of \textsc{Eulerian Extension} with $n$ vertices. 
We can solve $G$ probabilistically using poly$(n+\ell)$ evaluations of $Q(\bar x, z)$, using elements of bit-length $O(\log n+\log \ell)$. 
\end{lemma}
\begin{proof}
By Theorem~\ref{thm:Q}, the terms of $Q(\bar x, z)$ correspond exactly to perfect matchings $M$ of $G_2$ such that $G+EP(M)$ is Eulerian. 
Furthermore, for every such matching $M$, the exponent of $z$ in the corresponding term $\prod_{e \in M} x(e)$ 
equals the weight of $EP(M)$ whenever the latter is at most $\ell$ (otherwise, both of them are larger that $\ell$). Thus, our task is to decide whether $Q(\bar x, z)$ contains a term where the degree of $z$ is at most $\ell$. 
Let $L=O(\ell n^2)$ denote the maximum possible degree of $z$. 
Then, observe that we may rewrite $Q(\bar x, z)$ as $Q(\bar x,z) = \sum_{i=0}^L \alpha_i(\bar x) \cdot z^i$ by grouping terms;
hence our task is to decide whether there is $i \leq \ell$ such that $\alpha_i \not \equiv 0$. Note that each $\alpha_i$ is a polynomial in $\bar x$
of total degree at most $n/2$ (since the terms $x(e)$ come from perfect matchings $M$).


Instantiate $\bar x$ randomly from GF$(2^r)$ for some $r= \Omega(\log \ell + \log n)$.
By Lemma \ref{lem:SZ} and the union bound, with probability at least $1-1/\text{poly}(\ell+n)$,
we have $\alpha_i(\bar x)=0$ for $i\leq \ell$ if and only if $\alpha_i \equiv 0$. 
This instantiation defines a univariate polynomial $p(z)$ of maximum degree $L$, where we want to
decide whether there is a non-zero coefficient for $z^i$ for some $i\leq \ell$; this can be done via 
Lemma \ref{lem:polynomial interpolation} (at the cost of $L+1$ evaluations of $Q(\bar x, z)$). 
\end{proof}

\subsection{The undirected case}

With minor appropriate modifications (in particular exchanging bipartite matching and determinants
by general matchings and Pfaffians, see later)
we can also solve URPP. 

We now briefly outline the modifications required to handle the (seemingly more general) \textsc{Undirected Rural Postman Problem} (URPP).
The solution is very close to that of DRPP, with matchings in general graphs replacing bipartite graphs, and with Pfaffians replacing determinants. 
Let us define the problem properly.

\begin{center}
 \begin{namedefn}%
   {{\sc Undirected Rural Postman Problem (URPP)}}%
   {A connected multigraph $G=(V,E)$, 
   a subset $R$ of edges of \newline  $E$, 
   a weight function $\omega: E \rightarrow \N$, and 
   an integer $\ell$.}%
   {Is there a closed walk on $G$ containing every edge of $R$\newline
   with the total weight at most $\omega(R)+\ell$, where $\omega(R)$\newline is the total weight of $R$?}%
 \end{namedefn}
\end{center}

Recall that an undirected graph is Eulerian if and only if it is connected and every vertex is of even degree.
As for DRPP, we will solve URPP by working with the \textsc{Eulerian Extension} interpretation. 
A multiset $E'$ over the set $\{uv: u,v\in V\}$ of unordered pairs of vertices is an
{\em Eulerian extension (EE)} for $G$ if $G+E'=(V,E+E')$ is an Eulerian multigraph.

\begin{center}
 \begin{namedefn}%
   {{\sc Undirected Eulerian Extension (UEE)}}%
   {An undirected multigraph $G=(V,R)$ with no isolated \newline  vertices,
   a symmetric weight function $\omega: V\times V \rightarrow \mathbb{N}$,  and 
   \newline an integer $\ell$.}%
   {Is there an EE of total weight of most $\ell$?}%
 \end{namedefn}
\end{center}

The following is not hard to show. 

\begin{proposition}
There is a polynomial-time reduction from URPP, where the graph $G[R]$ of the URPP instance has $k$ connected components,
to UEE with a graph $G'$ with $k$ connected components, and where the weight function $\omega'$ obeys the triangle inequality. 
\end{proposition}

Henceforth, let $(G=(V,R),\omega,\ell)$ be the resulting instance of UEE,
let $O \subseteq V$ be the vertices of odd degree in $G$, 
let $V=V_1 \cup \ldots \cup V_k$ be the partition of $V$ induced by the connected components,
and let $O=O_1 \cup \ldots \cup O_k$ be the corresponding partition of $O$ (i.e., $O_i=O \cap V_i$). 
As for DRPP, we will refer to this set of connected components as $[k]$, 
identifying component $G[V_i]$ with the index $i$. 

\subsection{The no cycles transform}

Similarly as in the directed case, we give a transformation such that an optimal solution can be decomposed purely into paths.
For this, create two new vertices $v'$ and $v''$ for every $v \in V$, 
add edges $vv'$ and $vv''$ to $R$, and let $\omega(v,v')=\omega(v,v'')=\omega(v',v'')=0$. 
For any other edge $uv$, treat $v'$ and $v''$ as $v$, e.g., $\omega(u,v')=\omega(u,v)$. 
Keep $\ell$ unchanged. 

\begin{lemma}
After the above transformation, there exists a minimum weight EE which can
be partitioned into paths with endpoints in $O$, and in which every vertex of $O$
is an end-vertex of a unique such path.
\end{lemma}
\begin{proof}
The proof proceeds as that of Lemma~\ref{lemma:nocycles}.
Let $G'$ and $G''$ denote $G$ before and after the above transformation, and let $E'$ be a minimum weight EE for $G'$.
Assume w.l.o.g. that $E'$ is minimal, i.e., no subset of $E'$ is an EE. 
Partition $E'$ into a collection $Q$ of paths and cycles as follows. As long as any cycles remain in $E'$,
extract one cycle and put it in $Q$. Thereafter, when $E'$ is acyclic, repeatedly extract paths of maximum length,
e.g., leaf-leaf paths in the forest formed by the remaining edges of $E'$.
Observe that every vertex of $O$ is an end-vertex of a unique path in $Q$:
It is an end-vertex of at least one path 
since it has odd degree (the parity of a vertex $x$ is not changed when cycles through $x$ or paths in which $x$ is not an end-vertex, are deleted), 
and at most one path since every path is a leaf-leaf path at the time of its extraction. 

As in the proof of Lemma~\ref{lemma:nocycles}, we can show that the cycles of $Q$ admit a set of distinct representatives
(the proof of this goes through unchanged). We create an EE $E''$ for $G''$ as follows. 
For every cycle $C$ of $Q$, with representative vertex $v$, we create a path $P$ from $v'$ to $v''$, of the same cost as $C$, 
and add it to $E''$. For every path $P$ of $Q$, we simply add the edges of $P$ to $E''$. Finally,
for any vertex $v$ such that $v'$ and $v''$ still have odd degree, add the edge $v'v''$ to $E''$ (of weight 0). 
It is easy to see that $G''+E''$ forms an Eulerian multigraph, and that the weight of $E''$ equals the weight of $E'$. 
In the other direction, it is also easy to transfer any (arbitrary) EE for $G''$ into an EE for $G$ 
of the same weight. 
\end{proof}

The following still holds. 

\begin{lemma}
Assume a  minimum weight EE $E'$ is partitioned into paths as above. Let $P$ be
a path of the partition, starting at a vertex $u$, terminating at a vertex
$v$, and passing through a set $I \subseteq [k]$ of components. Let $P'$
be a minimum weight path from $u$ to $v$, subject to the constraint that $P'$ is
incident to each component $i \in I$. Then a multiset $E''$ obtained from $E'$ 
by replacing the edges of $P$ by the edges of $P'$ is a minimum weight EE, too.
\end{lemma}

Now, for every $u, v \in O$ with $u \in V_i$ and $v \in V_j$ (where we may have $i=j$),
and for every $I$ with $ \{i,j\} \subseteq I \subseteq [k] $, let $P(u,I,v)$ represent an (arbitrary, but fixed) 
minimum weight path from $u$ to $v$ passing through the component set $I$;
again we may assume that $P(u,I,v)$ contains exactly $|I|$ vertices.
We assume that $P(u,I,v)$ and $P(v,I,u)$ refer to the same path. 


\subsection{Matching form}

Let $G_O$ refer to the weighted multigraph on vertex set $O$, where for every distinct
path $P(u,I,v)$ created in the end of the previous subsection there is an edge $e_{u,I,v}$ between $u$ and $v$ in $G_O$ representing the path. 
The weight of $e_{u,I,v}$ equals the weight of the path $P(u,I,v)$ if it is at most $\ell$, otherwise it equals $\ell+1$. 

Given a perfect matching $M$ in $G_O$, we define a multiset $EP(M)$ of edges as follows. 
Initially $EP(M)=\emptyset$, then for every edge $e_{u,I,w}\in M$ we add the edges of the path $P(u,I,w)$ to $EP(M)$,
creating multiple copies of some edges, if necessary. The following is a simple observation.

\begin{lemma}\label{lem:UEM}
An instance of {\sc UEE} is positive if and only if there is a perfect matching
$M$ in the above multgraph such that $G+EP(M)$ is connected and the total weight of $EP(M)$ is at most $\ell.$
\end{lemma}

\subsection{Sieving for connected solutions}
\label{sec:pfaffian}

For the next step, we will need to compute the \emph{Pfaffian} of a matrix; let us briefly recall some definitions. 
Let $B=[b_{ij}]$ be a skew-symmetric (defined as $B^T=-B$) $2q\times 2q$ matrix over a field $\mathbb{F}$. For each partition $P=\{\{i_1,j_1\},\ldots ,\{i_q,j_q\}\}$ of $[2q]$ into pairs, let $b_P=\sgn (\pi) b_{i_1j_1}\cdots b_{i_qj_q}$, where $\sgn (\pi)$ is the sign of permutation  $\pi=(i_1,j_1,\ldots ,i_q,j_q)$.
The  \emph{Pfaffian} of $B$, denoted $\pf B$, is the sum $\sum_P b_P$ over all partitions $P$ of $[2q]$ into pairs. For more information on Pfaffians of matrices, see \cite{LovPlu}.

Let $G=(V,E)$ be a graph; for convenience let $V=[n]$. We may assume that $n$ is even.
The Tutte matrix of $G$ is a $|V| \times |V|$ matrix $A$ such that 
$A(i,j) = x_{i,j}$ if $i<j$ and $ij \in E$, $A(i,j)=-x_{j,i}$ if $i>j$ and $ji\in E$,
and $A(i,j)=0$ otherwise. The matrix $A$ is skew-symmetric. 
Throughout, all variables $x_{i,j}$ are distinct indeterminates.

Over a field of characteristic two, we can drop the signs in the definition of the Pfaffian, and get
\[
\pf A = \sum_M \prod_{e \in M} x_e,
\]
where $M$ ranges over all perfect matchings of $G$.

It is well-known that the Pfaffian of a matrix can be computed in polynomial time (see, e.g., \cite{Rote2001}).

We now proceed with our algorithm.
For every edge $e_{u,I,v}$ of $G_O$, let $x(e_{u,I,v})=x_{u,I,v}z^r$ where $x_{u,I,v}$ is a new indeterminate,
$z$ is a common indeterminate, and $r$ equals the weight of $e_{u,I,w}$.
For every $I \subseteq [k]$, denote $V_I=\bigcup_{i \in I} V_i$, 
and let $A_I$ be a matrix over a field of characteristic two, with rows and columns
indexed by $O \cap V_I$, and with entries
$
A_I(u,v)=\sum_{J\subseteq I} x(e_{u,J,v}).
$
Recall that the graph is undirected, i.e., $e_{u,I,v}$ and $e_{v,I,u}$ are considered as the same edge,
hence $x_{u,I,v}$ and $x_{v,I,u}$ are the same variables, and $A_I$ is a symmetric matrix for each $I$. 
Furthermore, since $A_I$ is over characteristic two, it is also skew-symmetric (i.e., $A_I^T=-A_I$)
with a zero diagonal. Hence, given an instantiation of $A_I$
we can compute the Pfaffian of $A_I$ in polynomial time. 
Now define 
\[
Q(\bar x, z):=\sum_{I \subseteq [k]\setminus\{1\}} (\pf{A_I})(\pf{A_{[k] \setminus I}}).
\]
We let $\pf A_{\emptyset}=1$. 
As before, we show that $Q(\bar x,z)$ enumerates exactly those perfect matchings $M$
for which $G+EP(M)$ is connected. 

The following lemma is immediate from the definition of the Pfaffian. 

\begin{lemma}
\label{lemma:U1}
Let $I \subseteq [k]$. Let $G_{O,I}=(O \cap V_I, E_I)$
denote the multigraph with edges $e_{u,J,w}$ from $G_O[V_I]$ for which $J \subseteq I.$
Then 
\[
\pf A_I=\sum_{M_I} \prod_{e \in M_I} x(e), \mbox{ where $M_I$ ranges over perfect matchings of $G_{O,I}.$}
\]
\end{lemma}

The remaining proofs proceed almost exactly as for the directed case.

\begin{lemma}\label{lem:uonly-matchings}
Every term of $Q(\bar x,z)$ equals $\prod_{e \in M} x(e)$ for some perfect
matching $M$ of $G_O$, where $x(e)$ for $e=e_{u,I,v}$ equals $x(e_{u,I,v})$.
\end{lemma}
\begin{proof}
Let $I \subseteq [k]\setminus \{1\}$. By Lemma~\ref{lemma:U1}, 
each term in both $\pf A_I$ and $\pf A_{[k] \setminus I}$ corresponds to a perfect matching of
$G_{O,I}$ and $G_{O,[k] \setminus I}$, respectively. Let $M_1$ and $M_2$ be arbitrary perfect
matchings of the respective graphs. Since $V_I$ and $V_{[k] \setminus I}$
partition $V$, we have that $M_1$ and $M_2$ are disjoint and 
$M=M_1 \cup M_2$ forms a perfect matching of $G_O$. Furthermore, 
$(\prod_{e \in M_1} x(e))(\prod_{e \in M_2} x(e))=\prod_{e \in M} x(e)$.
Thus for every selection $I$, each term of $Q$ generated equals 
$\prod_{e \in M} x(e)$ for some perfect matching $M$, and $Q(\bar x,z)$
as a whole thus enumerates perfect matchings in this sense.
\end{proof}

\begin{lemma}\label{lem:udiscon}
Let $M$ be a perfect matching of $G_O$ such that $G+EP(M)$ has $\rho$ connected components. 
Then the term $\prod_{e \in M} x(e)$ is enumerated $2^{\rho-1}$ times
in $Q(\bar x, z)$. 
\end{lemma}
\begin{proof}
Let $\rho=1$. Clearly $M$ is generated in the term corresponding to $I=\emptyset$ (in
the second part of the product). We show that no other choice of $I$
generates $M$. Let $I \subseteq [k]\setminus \{1\}$ be non-empty. Since $G+EP(M)$ is
connected, there is some edge $e_{u,J,v}$, $J\subseteq I$, in $M$ such that either $u$ and
$v$ lie on different sides of the partition $(V_I, V_{[k]\setminus I})$,
or $u,v$ lie on the same side but some $j \in J$ lies on the other side.
In both cases, the edge $e_{u,J,v}$ cannot be generated in either $\pf
A_I$ or $\pf A_{[k] \setminus I}$. 

Now we may assume that $\rho\ge 2$.
Let $C_1, \ldots, C_{\rho}$ be the connected components of $G+EP(M)$ and observe that 
for each $i\in [k]$, the vertex set $V_i$ is contained in a single component $C_j$. 
Let $[k]=I_1 \cup \ldots \cup I_{\rho}$ denote the partition of $[k]$ according to the components $C_i$,
i.e., $I_i=\{j \in [k]: V_j \subseteq V(C_i)\}$; choose the numbering
so that $1 \in I_1$.  For $J \subseteq [\rho]$, let $I_J:=\bigcup_{j \in J} I_j$. 

Let $1\in J\subseteq [\rho]$.
Observe that $M$ partitions into one perfect
matching $M_1$ for $G_{O,I_J}$ and one perfect matching $M_2$ for 
$G_{O, [k] \setminus I_J}$.
Since $\prod_{e \in M_1} x(e)$ is generated in $\pf A_{I_J}$, and $\prod_{e \in M_2} x(e)$ is generated in $\pf
A_{[k] \setminus I_J}$, $\prod_{e \in M} x(e)$ is generated (exactly once)
in their product. By an argument similar to the one used in the first paragraph of this proof, we can see that
$M$ can be generated only as above. Since there are exactly $2^{\rho -1}$ ways to choose a subset $J$ of $[\rho]$ containing 1, we are done.
\end{proof}

Thus, we have the following:

\begin{theorem}\label{thm:QU}
Over a field of characteristic two, $Q(\bar x,z)$ enumerates exactly those perfect matchings $M$
for which $G+ EP(M)$ is weakly connected. 
\end{theorem}

\subsection{Fast evaluation}

The procedure for evaluating $Q(\bar x, z)$ in time $O^*(2^k)$, given an instantiation of the variables $\bar x$ and $z$,
is again close to that for the directed case. First, for $u, v \in O$ and $I \subseteq [k]$, let $d(u,I,v)$
denote the length of the path $P(u,I,v)$ (where we assume that $u$ and $v$ are not contained in $V_i$ for any $i \in I$). 
We observe that we can compute the values $d(u,I,v)$ for all such triples $(u,I,v)$ in time $O^*(2^k)$ -- in fact, 
we can simply reuse the procedure for the directed case, since we are currently dealing with the special case
when $\omega$ is symmetric. 
Next, we create a table for the concrete values of $x(e_{u,I,v})=x_{u,I,v}z^r$
using the values of $d(u,I,v)$ and the given values for $\bar x$ and $z$. 
Finally, for each pair $u,v \in O$, we simultaneously compute $A_I(u,v)$ for all $I \subseteq [k]$,
since $A_I(u,v)$ is the zeta transform of $x(e_{u,I,v})$. 
We thus find that we can evaluate $Q(\bar x, z)$ in time $O^*(2^k)$, given values for variables $\bar x$ and $z$. 
We can now finish the proof of the following result. (Note that by Theorem~\ref{thm:QU}, 
the problem that remains to be solved in Theorem~\ref{thm:main2U} is entirely algebraic,
and there is no distinction between the UEE and the EE case except the definition
of the underlying polynomial $Q$; hence we omit the proof.)

\begin{theorem}\label{thm:main2U}
The \textsc{UEE} problem with $k$ components in $G$ and $\ell$ bounded by a polynomial in the number of vertices in $G$
can be solved by a randomized algorithm in $O^*(2^k)$ time and space.
\end{theorem}

\section{Conjoining Bipartite Matching Problem}\label{sec:match}

For a polynomial $P$ and a monomial $M,$ we let $\coef_P M$ denote the coefficient of $M$ in $P.$ The next lemma provides a way to extract from $P$ only monomials divided by a certain term.

\begin{lemma}\label{lem:coef}\cite{Wahl2013}
Let $P(x_1, . . . , x_n)$ be a polynomial over a field of characteristic two, and $T \subseteq [n]$
a set of target indices. For a set $I \subseteq [n],$ define $P_{- I}(x_1,\ldots  , x_n) = P(y_1,\ldots  , y_n),$ where $y_i = 0$
for $i \in I$ and $y_i = x_i$, otherwise. Define
$$Q(x_1,\ldots , x_n) =\sum_{I\subseteq T} P_{- I}(x_1,\ldots , x_n).$$
Then for any monomial $M$ divisible by $\Pi_{i\in T} x_i$ we have $\coef_Q M = \coef_P M,$ and for
every other monomial we have $\coef_Q M = 0.$
\end{lemma}

We now proceed with the proof of the main result of this subsection.

\begin{theorem}
\label{theorem:cbm fpt}
Let $\ell$ be bounded by a polynomial in the number of vertices in $B$. Then CBM 
can be solved by a randomized algorithm in $O^*(2^{|F|})$ time and polynomial space.
\end{theorem}
\begin{proof}
Introduce a variable $x_e$ for every edge $e$ of $B$, a variable $y_f$ for every $f\in F$, and a variable $z$. Let $U \cup W$ be the bipartition of vertices of $B$. 
Let $n=|U|=|W|$. 
Construct a matrix $A$, with rows indexed by $U$ and columns by $W$, as follows. For $u \in U$ and $w \in W$, we let $y(u,w)=y_{ij}$ if $u \in V_i$ and $w \in V_j$ (or vice versa)
and $ij \in F$ (i.e., if including an edge from $u$ to $w$ would satisfy the request $ij \in F$); otherwise $y(u,w)=1$. 
Now let $A(u,w)=x_{uw} \cdot y(u,w) \cdot z^{\omega(u,w)}$ if $uw \in E$, where $\omega(u,w)$ is the weight of the edge $uw$, and $A(u,w)=0$, otherwise.  
We claim that there is a conjoining matching of total weight at most $\ell$ if and only if
$\det A$, viewed as a polynomial, contains a monomial where the degree of $z$ is at most $\ell$, and where every variable $y_f$, $f \in F$, occurs at least once. 

To prove this claim, observe that the non-zero monomials of $\det A$ are exactly the perfect matchings of $B$
(it is easy to see that all non-zero transversals of $A$ produce distinct monomials, by the use of distinct edge variables $x_e$).
By inspecting the construction, we also find that in each monomial $T$ of $\det A$, corresponding to a matching $M$, 
the degree of $z$ in $T$ is exactly the total weight of $M$, and that $y_{ij}$ occurs in $T$ if and only if $y(u,w)=y_{ij}$ for some $uw \in M$,
i.e., if and only if $M$ contains some edge $uw$ satisfying the request $ij \in F$. (Note that the variables $y_{ij}$ may well have degree more
than one in $T$.) This proves our claim.

It remains to test in the allotted time whether such a monomial exists in $\det A$. We do this in two phases. 
Let $\det A(X,Y,z)$ denote the value of $\det A$ for a given evaluation of variables $x_e$, $y_f$, and $z$. 
Further, for $S \subseteq F$, let $\det A_{-S}(X,Y,z)$ denote $\det A(X,Y',z)$ where $Y'(f)=0$ if $f\in S$ and $Y'(f)=Y(f)$, otherwise. 
For the first phase, we define a secondary polynomial
$
p(X,Y,z) = \sum_{S \subseteq F} \det A_{-S}(X,Y,z);
$
then, by Lemma \ref{lem:coef}, $p(X,Y,z)$ contains exactly those monomials where every $y_f, f \in F$ occurs
(recall that computations are over GF$(2^r)$ for some appropriate $r$).
Thus $p$ enumerates exactly those monomials $T$ of $\det A$ which satisfy all requests in $F$, regardless of solution weight. 
Observe that $p$ can be evaluated in $O^*(2^{|F|})$ time. 

For the second phase, we need to detect whether there is at least one monomial in $p$ where the degree of $z$ is $\ell$ or less. 
This can be done as in the proof of Theorem \ref{thm:main2}.
\end{proof}

Finally, let us remark that the restriction that $B$ is a bipartite graph can be removed, i.e., that the {\sc Conjoining General Matching} (CGM) problem can be solved with very similar methods.
The proof is identical to the one above, except with the Tutte matrix replacing the bipartite adjacency matrix and Pfaffians replacing determinants, as in Section~\ref{sec:pfaffian}. The details are omitted. 

\begin{theorem}
\label{theorem:cgm fpt}
The {\sc Conjoining General Matching} problem can be solved in randomized time $O^*(2^{|F|})$ and polynomial space. 
\end{theorem}

\section{Discussion}\label{sec:dis}

Our results on {\sc DRPP} parameterized by $k$ and CBM parameterized by $|F|$, raise two natural questions. 
\begin{itemize}
\item Is it possible to get rid of the conditions bounding $\ell$ by a polynomial in the number of vertices in the corresponding graph? It seems that this question is not easy to answer as the complexity of many algorithms using algebraic methods heavily depends on a condition similar to ours, see, e.g., \cite{CyGaSa2012,MaPi}. 

\item Can we get rid of randomness in the algorithms? This seems to be a difficult
question. For some of the problems that can be solved by the randomized algebraic techniques
of \cite{MuVaVa1987}, no deterministic polynomial-time algorithms have been found, despite significant efforts.
\end{itemize}



\end{document}